%% file: main_arxiv.tex
\documentclass{article}

\usepackage{amsmath,amsthm,amsfonts,mathtools,algorithmicx,dsfont,amssymb}
\usepackage[most]{tcolorbox}
\usepackage[utf8]{inputenc} 
\usepackage[T1]{fontenc}    
\usepackage{url}            
\usepackage{booktabs}       
\usepackage{nicefrac}       
\usepackage{microtype}      
\usepackage{algpseudocode}
\usepackage{placeins}
\usepackage{graphicx}
\usepackage{amsmath}
\usepackage{todonotes}
\usepackage[numbers]{natbib}
\usepackage{geometry}
\usepackage[colorlinks,linkcolor=magenta,citecolor=blue, pagebackref=true]{hyperref}

\newtheorem{theorem}{Theorem}
\newtheorem{proposition}{Proposition}

\newcommand{\EE}{\mathbb E}
\newcommand{\PP}{\mathbb P}
\newcommand{\1}{\mathds 1}

\newcommand{\Acal}{\mathcal A}
\newcommand{\Ccal}{\mathcal C}

\newcommand{\cond}{\Ccal}

\newcommand{\apK}{\mathrm{apK}}
\DeclareMathOperator{\CI}{\text{CI}}

\newtcolorbox[auto counter,number within=chapter]{algorithm}[2][]{
  enhanced,
  title={Algorithm \thetcbcounter: #2}
}

\algnewcommand\algorithmicinput{\textbf{Input:}}
\algnewcommand\algorithmicoutput{\textbf{Output:}}
\algnewcommand\Input{\item[\algorithmicinput]}%
\algnewcommand\Output{\item[\algorithmicoutput]}%

\usepackage{setspace}

\begin{document}
%

%
\title{RiLACS: Risk-Limiting Audits via Confidence Sequences}

\author{Ian Waudby-Smith$^1$, Philip B. Stark$^2$, and Aaditya Ramdas$^1$\vspace{0.1in}\\
  $^1$Carnegie Mellon University\\
  $^2$University of California, Berkeley \vspace{0.05in}\\
  \texttt{ianws@cmu.edu}, \texttt{stark@stat.berkeley.edu}, \texttt{aramdas@cmu.edu}}

\maketitle

\setcounter{tocdepth}{2}
\makeatletter
\renewcommand\tableofcontents{%
    \@starttoc{toc}%
}

\begin{abstract}
\input{abstract}
\end{abstract}

\tableofcontents

\input{introduction}

\input{outline_arxiv}

\input{confseq}

\input{canada_audit}

\input{tallies}

\input{conclusion}

\subsection*{Acknowledgements}
AR acknowledges support from the Block Center for technology and society, and from NSF DMS 1916320. Research reported in this paper was sponsored in part by the DEVCOM Army Research Laboratory under Cooperative Agreement W911NF-17-2-0196 (ARL IoBT CRA). The views and conclusions contained in this document are those of the authors and should not be interpreted as representing the official policies, either expressed or implied, of the Army Research Laboratory or the U.S. Government. The U.S. Government is authorized to reproduce and distribute reprints for Government purposes notwithstanding any copyright notation herein.

\bibliographystyle{unsrtnat}
\bibliography{references}

\appendix
\input{supplement}

\end{document}

%% file: abstract.tex
Accurately determining the outcome of an election is a complex task with many potential sources of error, ranging from software glitches in voting machines to procedural lapses to outright fraud. 
Risk-limiting audits (RLA) are statistically principled ``incremental'' hand counts that provide statistical assurance that reported outcomes accurately reflect the validly cast votes.  
We present a suite of tools for conducting RLAs using confidence sequences --- sequences of confidence sets which uniformly capture an electoral parameter of interest from the start of an audit to the point of an exhaustive recount with high probability. Adopting the SHANGRLA \cite{stark2019sets} framework, we design nonnegative martingales which yield computationally and statistically efficient confidence sequences and RLAs for a wide variety of election types.

%% file: introduction.tex
\section{Introduction}
\label{section:intro}
The reported outcome of an election may not match the validly cast votes for a variety of reasons, including 
software configuration errors, bugs, human error,
and deliberate malfeasance.
Trustworthy elections start with a trustworthy
paper record of the validly cast votes.
Given access to a trustworthy paper trail of votes,
a risk-limiting audit (RLA) can 
provide a rigorous probabilistic guarantee:
\begin{enumerate}
    \item If an initially announced assertion $\mathcal{A}$ about an election is \emph{false}, this will be corrected by the audit with high probability;
    \item If the aforementioned assertion $\mathcal{A}$ is \emph{true}, then $\Acal$ will be confirmed (with probability one).
\end{enumerate}

Here, an electoral assertion $\mathcal{A}$ is simply a claim about the aggregated votes cast (e.g. ``Alice received more votes than Bob'').
An auditor may wish to audit several claims: for example, whether the reported winner is correct or whether the margin of victory is as large as announced.

From a statistical point of view, efficient risk-limiting audits 
can be implemented as sequential hypothesis tests. 
Namely, one tests the null hypothesis $H_0$: ``the assertion $\mathcal{A}$ is false,'' versus the alternative $H_1$: ``the assertion $\mathcal{A}$ is true''. 
Imagine then observing a random sequence of voter-cast ballots $X_1, X_2, \dots, X_N$, where $N$ is the total number of ballots.
A sequential hypothesis test is represented by a sequence $(\phi_t)_{t=1}^N$ of binary-valued functions:
\[ \phi_t := \phi(X_1, \dots, X_t) \mapsto \{ 0, 1\}, \]
where $\phi_t = 1$ represents rejecting $H_0$ (typically in favor of $H_1$), and $\phi_t = 0$ means that $H_0$ has not yet been rejected. The sequential test (and thus the RLA) stops as soon as $\phi_t = 1$ or once all $N$ ballots are observed, whichever comes first. 
The ``risk-limiting'' property of RLAs states that if the assertion is false (in other words, if $H_0$ holds), then
\[ \PP_{H_0}\left (\exists t \in \{1, \dots, N\} : \phi_t = 1 \right ) \leq \alpha, \]
which is equivalent to type-I error control of the sequential test.
Another way of interpreting the above statement is as follows:
if the assertion is incorrect, then with probability at least $(1-\alpha)$, $\phi_t = 0$ for every $t \in \{1, \dots, N\}$ and hence all $N$ ballots will eventually be inspected, at which point the ``true'' outcome (which is the result of the full hand count) will be known with certainty.

\subsection{SHANGRLA Reduces Election Auditing to Sequential Testing}

Designing the sequential hypothesis test $(\phi_t)_{t=1}^N$ depends on the type of vote, the aggregation method, or the social choice function for the election, and thus past works have constructed a variety of tests. 
Some works have designed $(\phi_t)_{t=1}^N$ in the context of a particular type of election \cite{lindeman2012bravo,ottoboni2019bernoulli,rivest2017clipaudit}. On the other hand, the ``SHANGRLA'' (\textbf{S}ets of \textbf{H}alf-\textbf{A}verage \textbf{N}ulls \textbf{G}enerate \textbf{RLA}s) framework unifies many common election types including plurality elections, approval voting, ranked-choice voting, and more by reducing each of these to a simple hypothesis test of whether a finite collection of finite lists of bounded numbers has mean $\mu^\star$ at most 1/2 \cite{stark2019sets,blom2021}. 
Let us give an illustrative example to show how SHANGRLA can be used in practice. 

Suppose we have an election with two candidates, Alice and Bob.
A ballot may contain a vote for Alice or for Bob, or it may contain no valid vote, e.g., because there was no selection or an overvote.
It is reported that Alice and Bob received $N_A$ and $N_B$ votes respectively with $N_A > N_B$ and that there were a total of $N_I$ invalid ballots for a total of $N = N_A + N_B + N_I$ voters. 
We encode votes for Alice as ``1'', votes for Bob as ``0'' and invalid votes as ``1/2'', to obtain a set of numbers $\{x_1, x_2, \dots, x_N\}$. 
Crucially, Alice indeed received more votes than Bob if and only if $\mu^\star := \frac{1}{N} \sum_{i=1}^N x_i > 1/2$. 
In other words, \emph{the report that Alice beat Bob can be translated into the assertion that $\mu^* \in (1/2,1]$}. 

SHANGRLA proposes to audit an assertion by testing its complement:
rejecting that ``complementary null'' is affirmative evidence that
the assertion is indeed true.
In other words, if one can ensure that $X_1, X_2, \dots, X_N$ is a random permutation of $\{x_1, \dots, x_N\}$ by sampling ballots without replacement (each ballot is chosen uniformly amongst remaining ballots), then we can concern ourselves with designing a hypothesis test $(\phi_t)_{t=1}^N$ to test the null ${H_0 : \mu^\star \leq 1/2}$ against the alternative ${H_1 : \mu^\star > 1/2}$.

One of the major benefits of SHANGRLA is the ability to reduce a wide range of election types to a testing problem of the above form. This permits the use of powerful statistical techniques which were designed specifically for such testing problems (but may not have been designed with RLAs in mind). Throughout this paper, we adopt the SHANGRLA framework, and while we return to the example of plurality elections for illustrative purposes, all of our methods can be applied to any election audit which has a SHANGRLA-like testing reduction \cite{stark2019sets}. 

\subsection{Confidence Sequences}
\label{section:introCS}
In the fixed-time (i.e. non-sequential) hypothesis testing regime, there is a well-known duality between hypothesis tests and confidence intervals for a parameter $\mu^\star$ of interest. 
We describe this briefly for $\mu^\star \in [0,1]$ for simplicity. For each $\mu \in [0,1]$, suppose that
\(
\phi^\mu \equiv \phi^\mu(X_1, \dots, X_n) \mapsto \{0, 1\} 
\)
is a level-$\alpha$ nonsequential, fixed-sample test for the hypothesis $H_0: \mu^\star = \mu$ versus $H_1: \mu^\star \neq \mu$. Then, a nonsequential, fixed-sample $(1-\alpha)$ confidence interval for $\mu^\star$ is given by the set of all $\mu \in [0,1]$ for which $\phi^\mu$ does not reject, that is
\( \{\mu \in [0,1] : \phi^\mu = 0 \} . \)

As we discuss further in Section~\ref{section:CS}, an analogous duality holds for sequential hypothesis tests and time-uniform \emph{confidence sequences} (here and throughout the paper, ``time'' is used to refer to the number of samples so far, and need not correspond to any particular units such as hours or seconds). We first give a brief preview of the results to come. 
Consider a family of sequential hypothesis tests $\{(\phi_t^\mu)_{t=1}^N\}_{\mu \in [0,1]}$, meaning that for each $\mu$, $(\phi_t^\mu)_{t=1}^N$ is a sequential test for $\mu$. 
Then, the set of all $\mu$ for which $\phi_t^\mu = 0$,
\[ C_t := \{ \mu \in [0, 1] : \phi_t^\mu = 0 \} \]
forms a $(1-\alpha)$ \emph{confidence sequence} for $\mu^\star$, meaning that
\[ \PP(\exists t \in [N] : \mu^\star \notin C_t ) \leq \alpha, \]
where $[N]$ is used to denote the set $\{1, 2, \dots, N \}$. In other words, $C_t$ will cover $\mu^\star$ at \emph{every single} time $t$, except with some small probability $\leq \alpha$. 
Since $C_t$ is typically an interval $[L_t, U_t]$, we call the lower endpoint $(L_t)_{t=1}^N$ as a lower confidence sequence (and similarly for upper).

\begin{figure}[!htbp]
    \centering
    \includegraphics[width=0.7\textwidth]{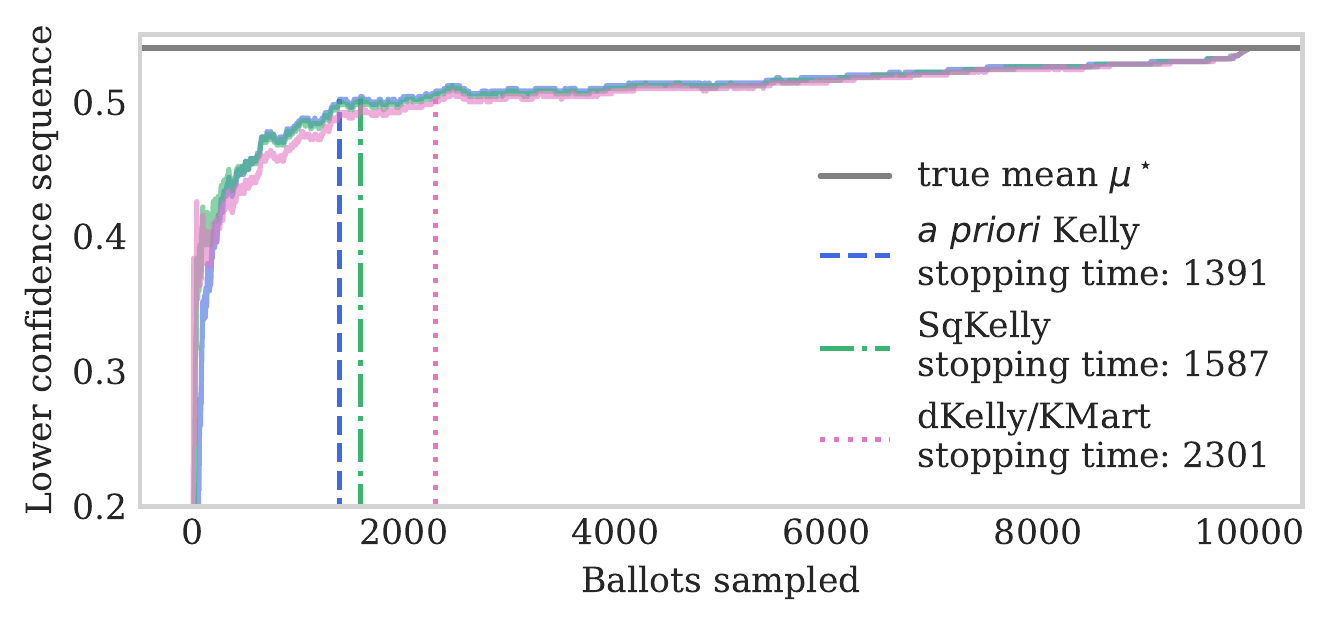}
    \caption{95\% Lower confidence sequences for the margin of a plurality election between Alice and Bob for three different auditing methods. Votes for Alice are encoded by ``1'' and those for Bob are encoded by ``0''. The parameter of interest is then the average of these votes, which in this particular example is 54\% (given by the horizontal grey line). The outcome is verified once the lower confidence sequence exceeds 1/2. The time at which this happens is given by the vertical blue, green, and pink lines.}
    \label{fig:exampleCS}
\end{figure}

In particular, given the sequential hypothesis testing problem that arises in SHANGRLA, we can cast the RLA as a sequential estimation problem that can be solved by developing confidence sequences (see Figure~\ref{fig:exampleCS}).\footnote{Code to reproduce all plots can be found at \href{https://github.com/wannabesmith/RiLACS}{github.com/wannabesmith/RiLACS}.} As we will see in Section~\ref{section:CS}, our confidence sequences provide added flexibility and an intuitive visualizable interpretation for SHANGRLA-compatible election audits, without sacrificing any statistical efficiency.

%% file: outline_arxiv.tex
\subsection{Contributions and Outline} 
The contributions of this work are twofold. First, we introduce confidence sequences to the election auditing literature as intuitive and flexible ways of interpreting and visualizing risk-limiting audits. Second, we present  algorithms for performing  RLAs based on confidence sequences by deriving statistically and computationally efficient nonnegative martingales. 
At the risk of oversimplifying the issue, modern RLAs face a computational-statistical efficiency tradeoff. Methods such as BRAVO are easy to compute, but potentially less statistically efficient than the current state-of-the-art, KMart \cite{stark2019sets}, but KMart can be prohibitively expensive to compute for large elections. The methods presented in this paper resolve this tradeoff: they typically match or outperform both BRAVO and KMart, while remaining practical to compute in large elections.

In Section~\ref{section:CS}, we show how confidence sequences generate risk-limiting audits, how they relate to more familiar RLAs based on sequentially valid $p$-values, and how they can be used to audit multiple contests. 
Section~\ref{section:design} derives novel confidence sequence-based RLAs and compares them to past RLA methods via simulation. 
In Section~\ref{section:canada}, we illustrate how the previously derived techniques can be applied to an audit of Canada's 43rd federal election. 
Finally, Section~\ref{section:tallies} discusses how all of the aforementioned results apply to risk-limiting tallies for coercion-resistant voting schemes.

%% file: confseq.tex
\section{Confidence Sequences are Risk-Limiting}
\label{section:CS}
Consider an election consisting of $N$ ballots. Following SHANGRLA \cite{stark2019sets}, suppose that these can be transformed to a set of $[0, u]$-bounded real numbers $x_1, \dots, x_N \in [0, u]$ with mean $\mu^\star := \frac{1}{N}\sum_{i=1}^N x_i$ for some known $u > 0$. Suppose that electoral assertions can be made purely in terms of $\mu^\star$. 
A classical $(1-\alpha)$ confidence interval $\CI_n$ for $\mu^\star$ is an interval computed from data $X_1, X_2, \dots, X_n$ with the guarantee that
\[ \forall n \in [N],\ \PP(\mu^\star \in \CI_n) \geq 1-\alpha. \]
In contrast, a $(1-\alpha)$ \emph{confidence sequence} for $\mu^\star$ is a sequence of confidence sets, $C_1, C_2, \dots, C_N$ which all simultaneously capture $\mu^\star$ with probability at least $(1-\alpha)$. That is,
\[ \underbrace{\PP(\forall t \in [N],\ \mu^\star \in C_t) \geq 1-\alpha}_{\text{simultaneous coverage probability}}, \quad \text{ or equivalently } \quad \underbrace{\PP(\exists t \in [N] : \mu^\star \notin C_t) \leq \alpha}_{\text{error probability}}. \]
The two probabilistic statements above are equivalent, but provide a different way of interpreting $\alpha$ and the corresponding guarantee.

If we have access to a $(1-\alpha)$ confidence sequence for $\mu^\star$, we can audit any assertion about the election outcome made in terms of $\mu^\star$ with risk limit $\alpha$. Here, we use $\Acal \subseteq [0, u]$ to denote an assertion. For example, SHANGRLA typically uses assertions of the form ``$\mu^\star$ is greater than $1/2$'', in which case $\Acal = (1/2, u]$.
\begin{algorithm}{Risk limiting audits via confidence sequences (RiLACS)}
\begin{algorithmic}[] 
    \Input Assertion $\Acal \subseteq [0, u]$, risk limit $\alpha \in (0, 1)$.
    \For{$t \in [N]$}
        \State Randomly sample and remove $X_t$ from the remaining ballots.
        \State Compute $C_t \equiv C(X_1, \dots, X_t)$ at level $\alpha$.
            \If{$\Acal \subseteq C_t$}
            \State Certify the assertion $\Acal$ and stop if desired.
        \EndIf
    \EndFor
\end{algorithmic}
\end{algorithm}
If the goal is to finish the audit as soon as possible above all else, then one can ignore the ``if desired'' condition. However, continued sampling can provide added assurance in $\Acal$, and maintains the risk limit at $\alpha$. The following theorem summarizes the risk-limiting guarantee of the above algorithm.
\begin{theorem}
\label{theorem:cs_rla}
    Let $(C_t)_{t=1}^N$ be a $(1-\alpha)$ confidence sequence for $\mu^\star$. Let $\Acal \subseteq [0, u]$ be an assertion about the electoral outcome (in terms of $\mu^\star$). The audit mechanism that certifies $\Acal$  as soon as $C_t \subseteq \Acal$ has risk limit $\alpha$.
\end{theorem}
\begin{proof}
    We need to prove that if $\mu^\star \notin \Acal$, then $\PP(\exists t \in [N] : C_t \subseteq \Acal) \leq \alpha$. First, notice that if $C_t \subseteq \Acal$, then we must have that $\mu^\star \notin C_t$ since $\mu^\star \notin \Acal$. Then,
    \begin{align*}
        \PP(\exists t \in [N] : C_t \subseteq \Acal) &\leq \PP(\exists t \in [N] : \mu^\star \notin C_t) \\
        &\leq \alpha,
    \end{align*}
    where the second inequality follows from the definition of a confidence sequence. This completes the proof. 
\end{proof}

Let us see how this theorem can be used in an example. Consider an election with two candidates, Alice and Bob, and a total of $N$ 
cast ballots. 
Let $\{x_1, \dots, x_N\}$ be the list of numbers that result from encoding votes for Alice as 1, votes for Bob as 0, and ballots that do not
contain a valid vote as $1/2$. 
Let $(C_t)_{t=1}^N$ be a $(1-\alpha)$ confidence sequence for $\mu^\star := \frac{1}{N} \sum_{i=1}^N x_i$. 
If we wish to audit the assertion that ``Alice beat Bob'', then $u=1$
and $\Acal = (1/2, 1]$. 
We can sequentially sample $X_1, X_2, \dots, X_N$ without replacement, certifying the assertion once $C_t \subseteq \Acal$. By Theorem~\ref{theorem:cs_rla}, this limits the risk to level $\alpha$.

\subsection{Relationship to Sequential Hypothesis Testing}
\label{section:relationshipToTesting}
The earliest work on RLAs did not use anytime $p$-values \cite{stark2008conservative,stark2009cast}, 
but
since about 2009, most RLA methods have used anytime $p$-values to conduct sequential hypothesis tests
\cite{stark2009risk,ottoboni2018risk,ottoboni2019bernoulli,stark2019sets,huang2020unified}. 
An anytime $p$-value is a sequence of $p$-values $(p_t)_{t=1}^N$ with the property that under some null hypothesis $H_0$,
\begin{equation}
\label{eq:anytime_pval}
     \PP_{H_0} ( \exists t \in [N] : p_t \leq \alpha ) \leq \alpha. 
\end{equation}
The anytime $p$-values $p_t \equiv p_t(\mu)$ are typically defined implicitly for each null hypothesis $H_0 : \mu^\star = \mu$ and yield a sequential hypothesis test $\phi_t^\mu := \1(p_t(\mu) \leq \alpha)$. 
As alluded to in Section~\ref{section:introCS}, this immediately recovers a confidence sequence:
\[ C_t := \{ \mu \in [0, u] : \phi_t^\mu = 0 \}. \]
Notice in Figure~\ref{fig:CSvsPval_ballots} that the times at which nulls are rejected (or ``stopping times'') are the same for both confidence sequences and the associated $p$-values. 
Thus, nothing is lost by basing
the RLA on confidence sequences rather than anytime $p$-values. 
Confidence sequences benefit from being visually intuitive and are arguably easier to interpret than anytime $p$-values.

For example, consider conducting an RLA for a simple two-candidate election between Alice and Bob with no invalid votes. 
Suppose that it is reported that Alice won, i.e., $\mu^\star := \frac{1}{N} \sum_{i=1}^N x_i > 1/2$ where $x_i = 1$ if the $i$th ballot is for Alice, 0 if for Bob, and $1/2$ if the ballot does not contain a valid vote for either candidate. 
A sequential RLA in the SHANGRLA framework would posit a null hypothesis $H_0 : \mu^\star \leq 1/2 $ (the complement of the announced result: Bob actually won or the outcome is a tie), sample random ballots sequentially, and stop the audit (confirming the announced result) if and when $H_0$ is rejected at significance level $\alpha$.
If $H_0$ is not rejected before all ballots have been inspected, the true outcome is known.\footnote{%
At any point during the sampling, an election
official can choose to abort the sampling and perform a full hand count for any reason. This cannot
increase the risk limit: the chance of failing to correct an incorrect reported outcome does not increase.
}

\begin{figure}[!htbp]
    \centering
    \includegraphics[width=0.95\textwidth]{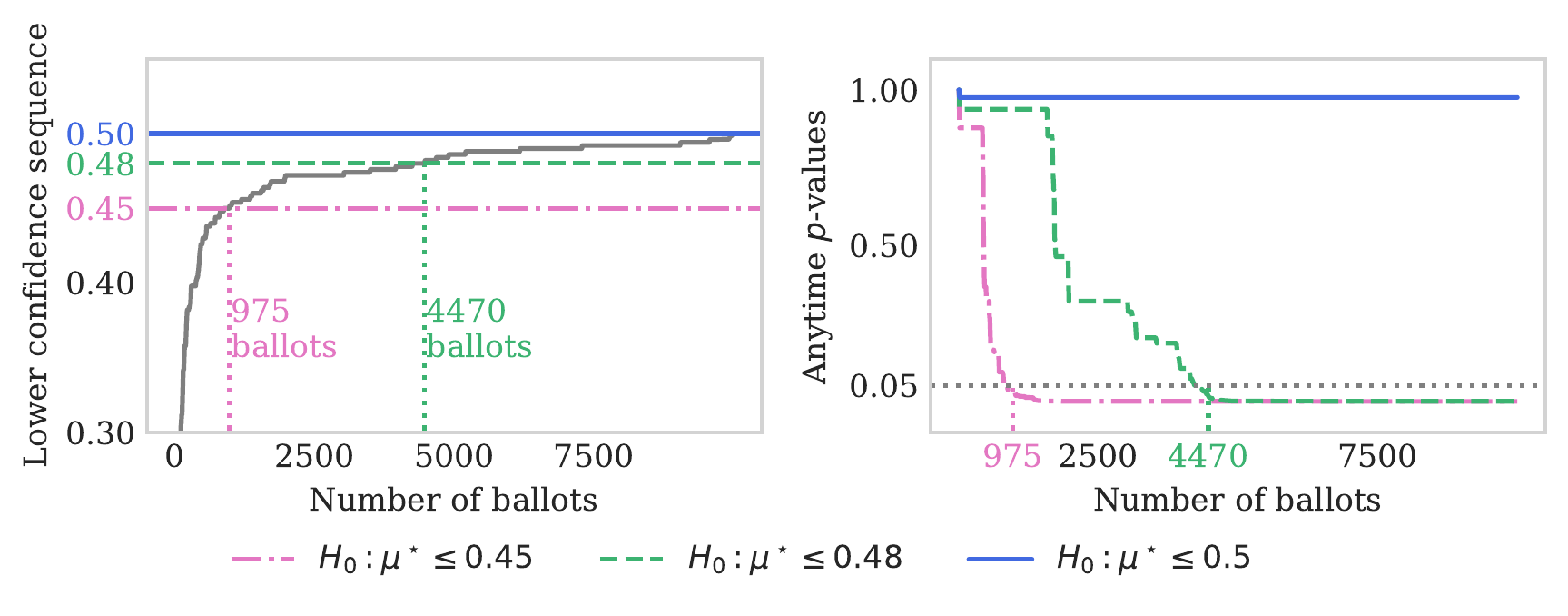}
    \caption{The duality between anytime $p$-values and confidence sequences for three nulls: $H_0: \mu^\star \leq \mu_0$ for $\mu_0 \in \{0.45, 0.48, 0.5\}$. The $p$-value for $H_0: \mu^\star \leq 0.45$ (pink dash-dotted line) drops below $5\%$ after 975 samples, exactly when the $95\%$ \emph{lower} confidence sequence exceeds 0.45. However, the $p$-value for $H_0 : \mu^\star \leq 0.5$ never reaches $0.05$ and the 95\% confidence sequence never excludes 0.5, the true value of $\mu^\star$.}
    \label{fig:CSvsPval_ballots}
\end{figure}

On the other hand, a ballot-polling RLA 
\cite{lindeman2012bravo} based on confidence sequences proceeds by computing a lower $1-\alpha$
confidence bound for the fraction $\mu^\star$ of votes for Alice.
The audit stops, confirming the outcome, if and when this lower bound is larger than 1/2.
If that does not occur before the last ballot has been examined, the true outcome is known.
In this formulation, there is no need to define a null hypothesis as the complement of the announced result and interpret the resulting $p$-value, and so on. 
The approach also works for comparison audits using the ``overstatement assorter'' approach developed in \cite{stark2019sets}, which
transforms the problem into the same canonical form: testing whether the mean of any list in a collection of nonnegative, bounded lists is less than 1/2.

\subsection{Auditing Multiple Contests}
\label{section:multipleContests}
It is known that RLAs of multi-candidate, multi-winner elections can be reduced to several pairwise contests without adjusting for multiplicity \cite{lindeman2012bravo}. This is accomplished by testing whether every single reported winner beat every single reported loser, and stopping once each of these tests rejects their respective nulls at level $\alpha \in (0, 1)$. For example, suppose it is reported that a set of candidates $\mathcal{W}$
beat
a set of candidates $\mathcal{L}$
in a $k$-winner plurality contest
with $K$ candidates in all
(that is, $|\mathcal{W}| = k$
and $|\mathcal{L}| = K-k$).
For each reported winner $w \in \mathcal{W}$
and each reported loser $\ell \in \mathcal{L}$,
encode votes for candidate $w$ as ``1'', votes for $\ell$ as ``0'' and
ballots with no valid vote in the contest or with a vote for any other candidate as ``1/2'' to obtain the population $\{x_{1}^{w, \ell}, \dots, x_{N}^{w, \ell}\}$. 
Then as before, candidate $w$ beat candidate $\ell$ if and only if $\mu^\star_{w,\ell} := \frac{1}{N} \sum_{i=1}^N x_{i}^{w, \ell} > 1/2$. 
In a two-candidate plurality election we would have proceeded by testing the null $H_{0}^{w, \ell} : \mu^\star_{w,\ell} \leq 1/2$ against the alternative $H_{1}^{w, \ell} : \mu^\star_{w,\ell} > 1/2$. 
To use the decomposition of a 
single winner or multi-winner plurality contest into a set of pairwise contests,
we test each null $H_{0}^{w, \ell} : \mu^\star_{w,\ell} \leq 1/2$ for 
$w \in \mathcal{W}$
and $\ell \in \mathcal{L}$.
The audit stops if and when \emph{all $k(K-k)$ null hypotheses} are rejected. 
Crucially, if candidate $w \in \mathcal{W}$ did not win (i.e. $\mu^\star_{w,\ell} \leq 1/2$ for some $\ell \in \mathcal{L}$), then
\[ \PP(\text{reject all } H_{0, w, \ell}: w \in \mathcal{W}, \ell \in \mathcal{L}) \leq \min_{w \in \mathcal{W}, \ell \in \mathcal{L}} \PP(\text{reject } H_{0, w, \ell}) \leq \alpha. \]
The same technique applies when auditing with confidence sequences.
Let $\{(C_{t}^{w, \ell})_{t=1}^N\}$
be $(1-\alpha)$ confidence sequences for
$\{\mu_{w, \ell}^*\}$,
$w \in \mathcal{W}$, $\ell \in \mathcal{L}$.
We verify the electoral outcome of every contest once $C_{t}^{w, \ell} \subseteq (1/2, u] $
for all
$w \in \mathcal{W}$, $\ell \in \mathcal{L}$.
Again, if $\mu_{w,\ell}^\star \leq 1/2$
for some $w \in \mathcal{W}$, and
$\ell \in \mathcal{L}$,
then
\begin{align*}
    &\PP(\forall w \in \mathcal{W}, \forall \ell \in \mathcal{L},\ C_{t} ^{w, \ell} \subseteq (1/2, u]) \leq \min_{w \in \mathcal{W}, \ell \in \mathcal{L} }
\PP(C_{t}^{w, \ell} \subseteq (1/2, u]) \leq \alpha. 
\end{align*}
This technique can be generalized to handle audits of any number of contests from the same audit sample, as explained in \cite{stark2019sets}. For the sake of brevity, we omit the derivation, but it is a straightforward extension of the above.


\section{Designing Powerful Confidence Sequences for RLAs}
\label{section:design}

So far we have discussed how to conduct RLAs from confidence sequences for the parameter $\mu^\star$. 
In this section, we will discuss how to derive powerful confidence sequences for the purposes of conducting RLAs as efficiently as possible. For mathematical and notational convenience in the following derivations, we consider the case where $u=1$. Note that nothing is lost in this setup since any population of $[0, u]$-bounded numbers can be scaled to the unit interval $[0, 1]$ by dividing each element by $u$ (thereby scaling the population's mean as well).

As discussed in Section~\ref{section:relationshipToTesting}, we can construct confidence sequences by ``inverting'' sequential hypothesis tests. In particular, given a sequential hypothesis test $(\phi_t^\mu)_{t=1}^N$, the sequence of sets,
\[ C_t := \{ \mu \in [0, 1] : \phi^\mu_t = 0 \} \]
forms a $(1-\alpha)$ confidence sequence for $\mu^\star$.
Consequently, in order to develop powerful RLAs via confidence sequences, we can simply focus on carefully designing sequential tests $(\phi_t^\mu)_{t=1}^N$.\footnote{Notice that it is not always feasible to compute the set of all $\mu \in [0, 1]$ such that $\phi_t^\mu = 0$ since $[0, 1]$ is uncountably infinite. However, all  confidence sequences we will derive in this section are intervals (i.e. convex), and thus we can find the endpoints using a simple grid search or standard root-finding algorithms.}

To design sequential hypothesis tests, we start by finding \emph{martingales} that translate to powerful tests. To this end, define $M_0(\mu) := 1$ and consider the following process for $t \in [N]$:
\begin{equation}
\label{eq:capitalProcess}
    M_t(\mu) := \prod_{i=1}^t \left (1 + \lambda_i (X_i - \cond_i(\mu)) \right ),
\end{equation}
where $\lambda_i \in \left [0, \tfrac1{\cond_i(\mu)} \right]$ is a tuning parameter depending only on $X_1, \dots, X_{i-1}$, and 
\[ \cond_i(\mu) := \frac{N \mu - \sum_{j=1}^{i-1} X_j}{N - i + 1} \]
is the conditional mean of $X_i \mid X_1, \dots, X_{i-1}$ if the mean of $\{x_1, \dots, x_N\}$ were $\mu$. 

Following \cite[Section 6]{waudby2020estimating}, the process $(M_t(\mu^\star))_{t=0}^N$ is a nonnegative martingale starting at one. 
Formally, this means that $M_0(\mu^\star) = 1$, $M_t(\mu^\star) \geq 0$, and 
\[
\EE(M_t(\mu^*) \mid X_1, \dots, X_{t-1}) = M_{t-1}(\mu^*) 
\]
for each $t \in [N]$. Importantly for our purposes, nonnegative martingales are unlikely to ever become very large. This fact is known as \emph{Ville's inequality} \cite{ville1939etude,howard_exponential_2018}, which serves as a generalization of Markov's inequality to nonnegative (super)martingales, and can be stated formally as
\begin{equation}
\label{eq:ville}
    \PP\left ( \exists t \in [N] : M_t(\mu^\star) \geq 1/\alpha \right ) \leq \alpha M_0(\mu^\star) = \alpha, 
\end{equation}
where $\alpha \in (0, 1)$, and the equality follows from the fact that $M_0(\mu^\star) = 1$. As alluded to in Section~\ref{section:CS}, $(M_t(\mu^\star))_{t=0}^N$ can be interpreted as the reciprocal of an anytime $p$-value:
\[ \PP\left (\exists t \in [N] : \frac{1}{M_t(\mu^\star)} \leq \alpha \right) \leq \alpha, \]
which matches the probabilistic guarantee in \eqref{eq:anytime_pval}.
As a direct consequence of Ville's inequality, if we define the test $\phi_t^\mu := \1(M_t(\mu) \geq 1/\alpha)$, then 
\[ \PP(\exists t \in [N] : \phi_t^{\mu^\star} = 1) \leq \alpha, \]
and thus $(\phi_t^{\mu})_{t=1}^N$ is a level-$\alpha$ sequential hypothesis test. We can then invert $(\phi_t^{\mu})_{t=1}^N$ and apply Theorem~\ref{theorem:cs_rla} to obtain confidence sequence-based RLAs with risk limit $\alpha$. 

\subsection{Designing Martingales and Tests from Reported Vote Totals}
So far, we have found a process $(M_t(\mu))_{t=0}^N$ that is a nonnegative martingale when $\mu = \mu^\star$, but what happens when $\mu \neq \mu^\star$? This is where the tuning parameters $(\lambda_t)_{t=1}^N$ come into the picture. Recall that an electoral assertion $\Acal$ is certified once $C_t \subseteq \Acal$. Therefore, to audit assertions quickly, we want $C_t$ to be as tight as possible. Since $C_t$ is defined as the set of $\mu \in [0, 1]$ such that $M_t(\mu) < 1/\alpha$, we can make $C_t$ tight by making $M_t(\mu)$ as \emph{large} as possible. To do so, we must carefully choose $(\lambda_t)_{t=1}^N$. This choice will depend on the type of election as well as the amount of information provided prior to the audit. First consider the case where reported vote totals are given (in addition to the announced winner).

\begin{figure}[h!]
    \centering
    \includegraphics[width=0.45\textwidth]{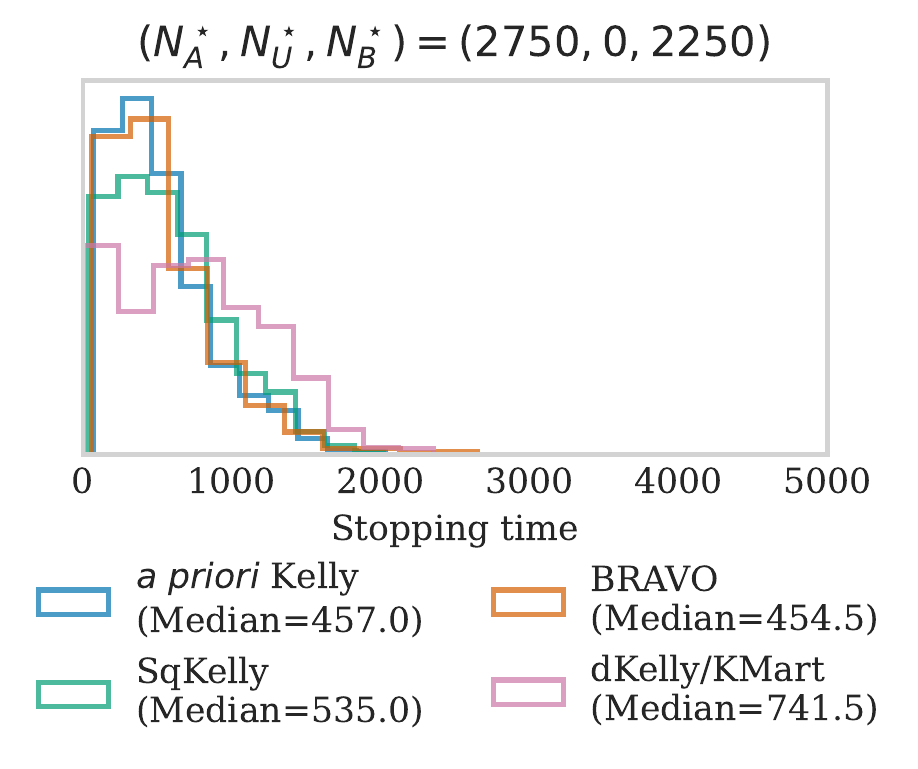}
    \centering 
    \includegraphics[width=0.45\textwidth]{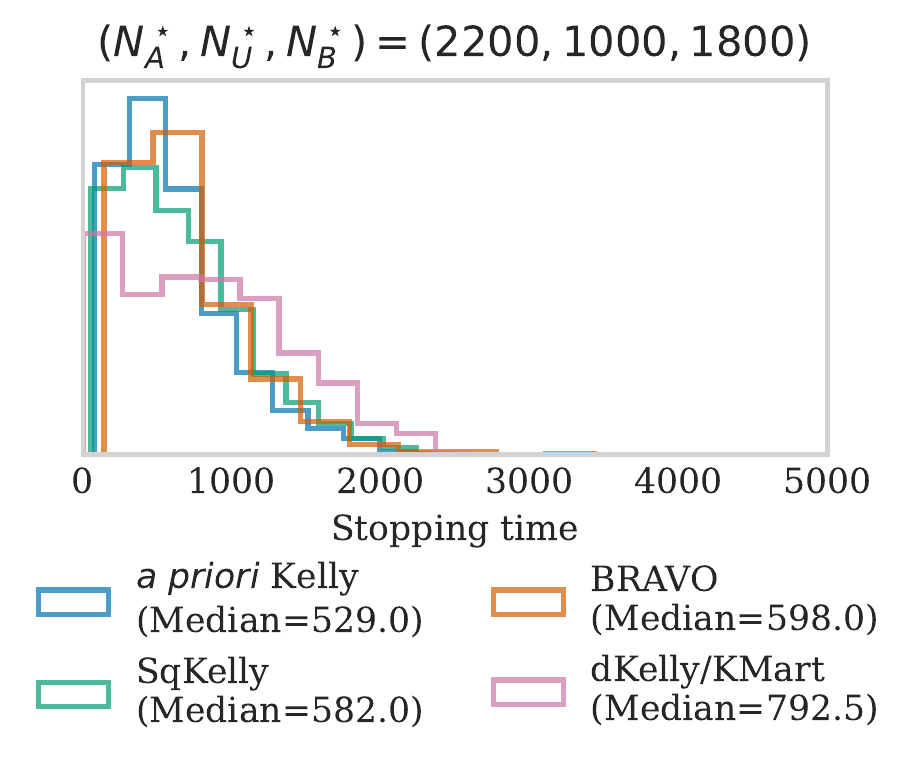}
    \centering 
    \includegraphics[width=0.45\textwidth]{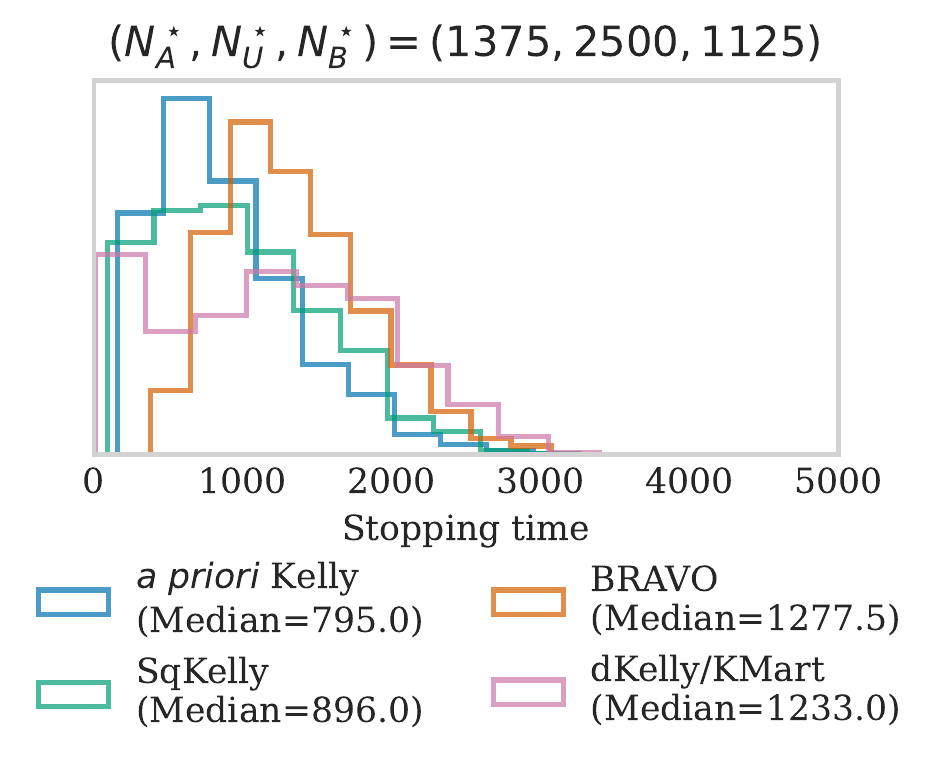}
    \centering
    \includegraphics[width=0.45\textwidth]{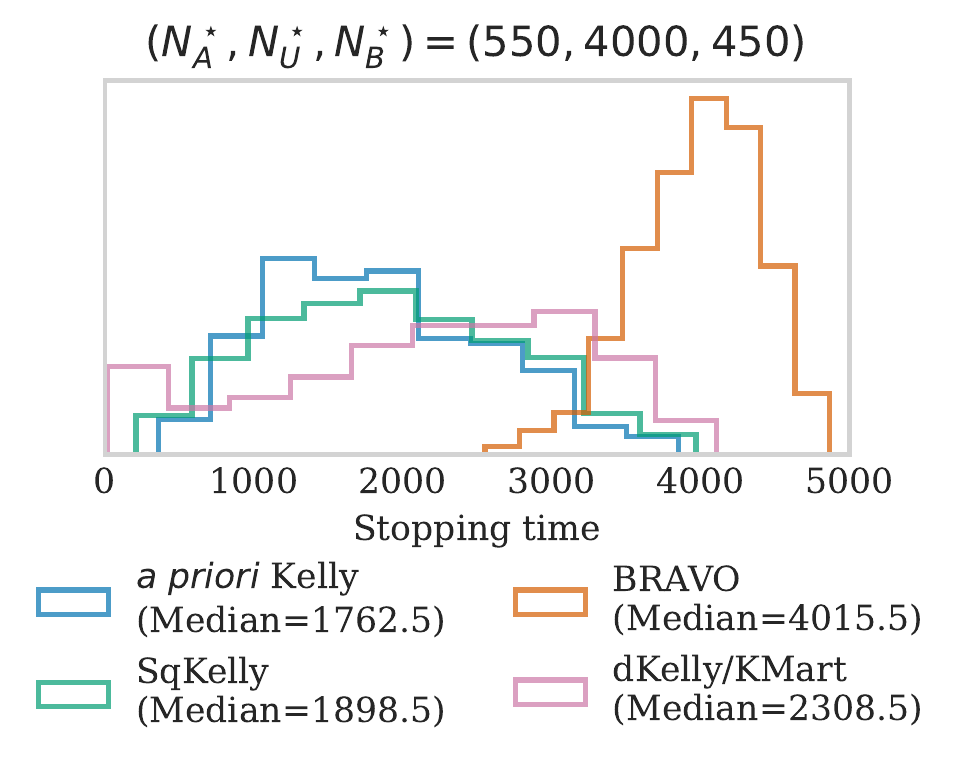}
    \caption{Ballot-polling audit workload distributions under four possible outcomes of a two-candidate plurality election. Workload is defined as the number of distinct ballots examined before completing the audit. The first example considers an outcome where Alice and Bob received 2750 and 2250 votes respectively, and no ballots were invalid, for a margin of $0.1$. The second, third, and fourth examples have the same margin, but with increasing numbers of invalid or ``nuisance'' ballots represented by $N_U^\star$. Notice that in the case with no nuisance ballots, \textit{a priori} Kelly and BRAVO have an edge, while in the setting with many nuisance ballots, \textit{a priori} Kelly vastly outperforms BRAVO. On the other hand, neither SqKelly nor dKelly require tuning based on the reported outcomes, but SqKelly outperforms dKelly in all four scenarios.}
    \label{fig:stoppingTimes}
\end{figure}

For example, recall the election between Alice and Bob of Section~\ref{section:CS}, and suppose that $\{x_1, \dots, x_N\}$ is the list of numbers encoding votes for Alice as 1, votes for Bob as 0, and ballots with no valid vote for either candidate as 1/2. 
Recall that Alice beat Bob if and only if $\mu^\star := \frac{1}{N} \sum_{i=1}^N x_i > 1/2$, so we are interested in testing the null hypothesis $H_0: \mu^\star \leq 1/2
\text{ 
against the alternative }
 H_1: \mu^\star > 1/2$.
Suppose it is reported that Alice beat Bob with $N_A'$ votes for Alice, $N_B'$ for Bob, and $N_U'$ nuisance votes (i.e. either invalid or for another party). If the reported outcome is \emph{correct}, then for any fixed $\lambda$, we know the exact value of 
\begin{equation}
\label{eq:finalWealth}
    \prod_{i=1}^N (1 + \lambda (x_i - 1/2) ),
\end{equation} 
which is an inexact but reasonable proxy for $M_N(1/2)$, the final value of the process $(M_t(1/2))_{t=0}^N$.
We can then choose the value of $\lambda'$ that maximizes \eqref{eq:finalWealth}. Some algebra (which we defer to Section~\ref{section:maximizeFinalWealth}) reveals that the maximizer of \eqref{eq:finalWealth} is given by
\begin{equation}
\label{eq:optimalLambda}
    \lambda' := 2\frac{N_A' - N_B'}{N_A' + N_B'}. 
\end{equation} 
We then truncate $\lambda'$ at each time step $t$ to obtain
\begin{equation}
\label{eq:apK}
     \lambda_t^\mathrm{apK} := \min \left \{ \lambda', \frac{1}{\cond_t(\mu)} \right \},
\end{equation}
ensuring that it lies in the allowable range $[0, 1/\cond_t(\mu)]$. We call this choice of $\lambda_t^\mathrm{apK}$ \textbf{\textit{a priori} Kelly} due to its connections to Kelly's criterion \cite{kelly1956new,waudby2020estimating} for maximizing products of the form \eqref{eq:finalWealth}. This choice of $\lambda_t^\apK$ also has the desirable property of yielding convex confidence sequences, which we summarize below.
\begin{proposition}
    \label{proposition:convex}
    Let $X_1, \dots, X_N$ be a sequential random sample from $\{x_1, \dots, x_N\}$ with $\mu^\star := \frac{1}{N} \sum_{i=1}^N x_i$. Consider $(\lambda_t^\apK)_{t=1}^N$ from \eqref{eq:apK} and define the process $M_t(\mu) := \prod_{i=1}^t (1 + \lambda_i^\apK(X_i - \cond_i(\mu)))$ for any $\mu \in [0, 1]$. Then the confidence set
    \[ C_t^\apK := \{ \mu \in [0, 1] : M_t(\mu) < 1/\alpha \} \]
    is an interval with probability one.
\end{proposition}
\begin{proof}
    Notice that since $\lambda' \geq 0,\ \cond_t(\mu) \geq 0,$ and $X_i \geq 0$, we have that
    \[\lambda_t^\apK(X_i - \cond_t(\mu)) = \min\{\lambda'X_i, X_i/\cond_t(\mu) \} - \min\{\lambda' \cond_t(\mu), 1\} \]
    is a nonincreasing function of $\mu$ for each $t \in [N]$. Consequently, $M_t(\mu)$ is a nonincreasing and quasiconvex function of $\mu$, so its sublevel sets are convex.
\end{proof}

Note that \emph{any} sequence $(\lambda_t)_{t=1}^N$ such that $\lambda_t \in [0, 1/\cond_t(\mu)]$ would have yielded a valid nonnegative martingale, but we chose that which maximizes \eqref{eq:finalWealth} so that the resulting hypothesis test $\phi_t := \1 ( M_t(1/2) > 1/\alpha)$ is powerful. In situations more complex than two-candidate plurality contests, the maximizer of \eqref{eq:finalWealth} can still be found efficiently via standard root-finding algorithms. All of these methods are implemented in our Python package.\footnote{\href{https://github.com/wannabesmith/rilacs}{github.com/wannabesmith/RiLACS}}

While audits based on \textit{a priori} Kelly display excellent empirical performance (see Figure~\ref{fig:stoppingTimes}), their efficiency may be hurt when vote totals are erroneously reported. Small errors in reported vote totals seem to have minor adverse effects on stopping times (and in some cases can be slightly beneficial), but larger errors can significantly affect stopping time distributions (see Figure~\ref{fig:misspecified_aPriori_Kelly}). If we wish to audit the reported winner of an election but prefer not to rely on (or do not have access to) exact reported vote totals, we need an alternative to \textit{a priori} Kelly. In the following section, we describe a family of such alternatives. 

\begin{figure}[!htbp]
    \centering
    \includegraphics[width=0.45\textwidth]{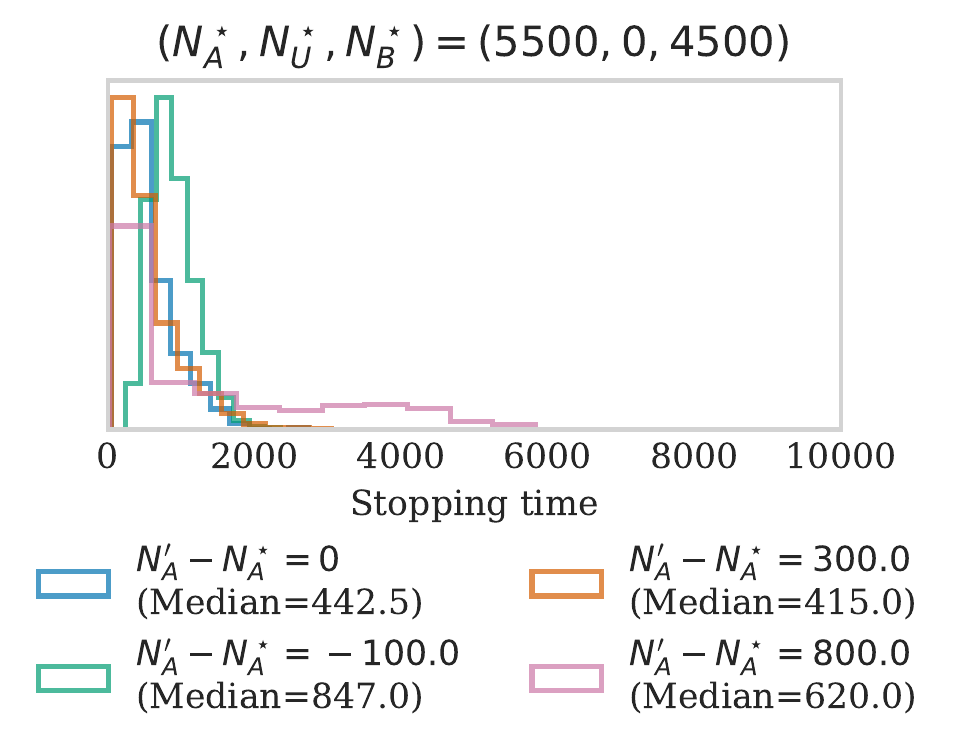}
    \centering 
    \includegraphics[width=0.45\textwidth]{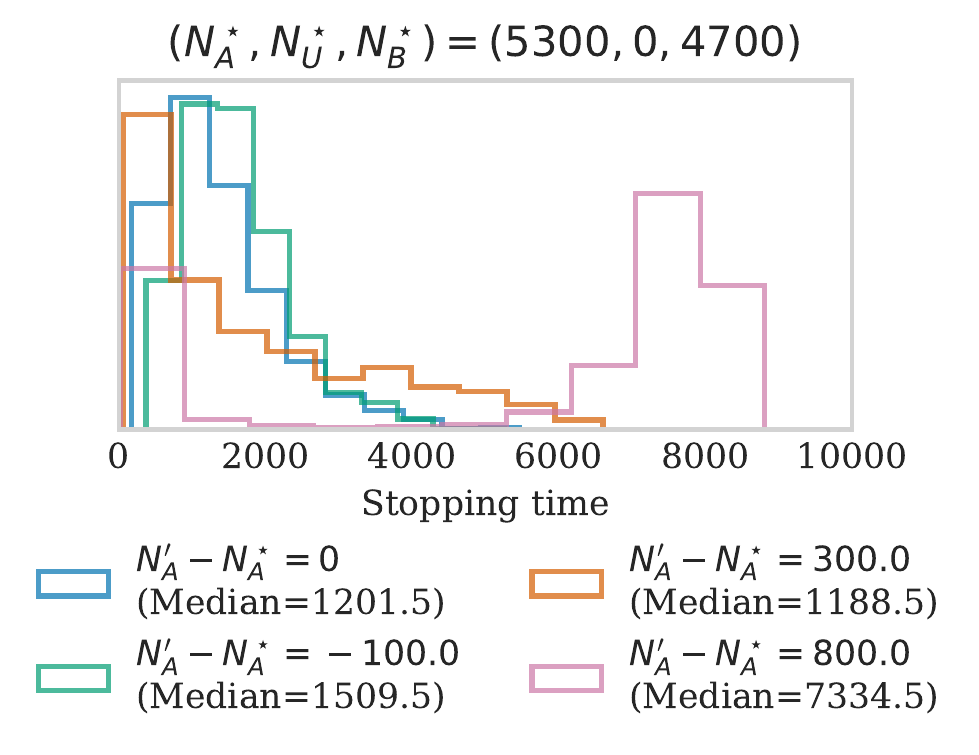}
    \centering
    \caption{Stopping times for \textit{a priori} Kelly under various degrees of error in reported outcomes. In the above legends, $N_A^\star$ refers to the \emph{true} number of votes for Alice, while $N_A'$ refers to the incorrectly reported number of votes. Notice that empirical performance is relatively strong for $N_A' - N_A^\star \in \{0, 300\}$ but is adversely affected when $N_A' - N_A^\star \in \{-100, 800\}$, especially in the right-hand side plot with a narrower margin.}
    \label{fig:misspecified_aPriori_Kelly}
\end{figure}

\subsection{Designing Martingales and Tests without Vote Totals}
\label{section:designWithoutTotals}
If the exact vote totals are not known, but we still wish to audit an assertion (e.g. that Alice beat Bob), we need to design a slightly different martingale that does not depend on maximizing \eqref{eq:finalWealth} directly. Instead of finding an optimal $\lambda'$, we will take $D \geq 2$ points evenly-spaced on the allowable range $[0, 1/\cond_t(\mu)]$ and ``hedge our bets'' among all of these. Making this more precise, note that a convex combination of martingales (with respect to the same filtration) is itself a martingale \cite{waudby2020estimating}, and thus for any $(\theta_1, \dots, \theta_D)$ such that $\theta_d \geq 0$ and $\sum_{d=1}^D \theta_d = 1$, we have that
\begin{equation}
\label{eq:diversifiedMartingale}
    M_t^D(\mu^\star) := \sum_{d=1}^D \theta_d \prod_{i=1}^t \left (1 + \frac{d}{(D+1) \cond_i(\mu^\star)} (X_i - \cond_i(\mu^\star)) \right) 
\end{equation} 
forms a nonnegative martingale starting at one. Notice that we no longer have to depend on the reported vote totals to begin an audit. Furthermore, confidence sequences generated using sublevel sets of $M_t^D(\mu)$ are intervals with probability one \cite[Proposition 4]{waudby2020estimating}. Nevertheless, choosing $(\theta_1, \dots, \theta_D)$ is a nontrivial task. A natural --- but as we will see, suboptimal --- choice is to set $\theta_d = 1/D$ for each $d \in [D]$. Previous works \cite{waudby2020estimating} call this \textbf{dKelly} (for ``diversified Kelly''), a name we adopt here. In fact, this choice of $(\theta_1, \dots, \theta_D)$ gives an arbitrarily close and computationally efficient approximation to the \emph{Kaplan martingale} (\textbf{KMart}) \cite{stark2019sets} which can otherwise be prohibitively expensive to compute for large $N$. 

\begin{figure}[!htbp]
    \centering
    \includegraphics[width=0.6\textwidth]{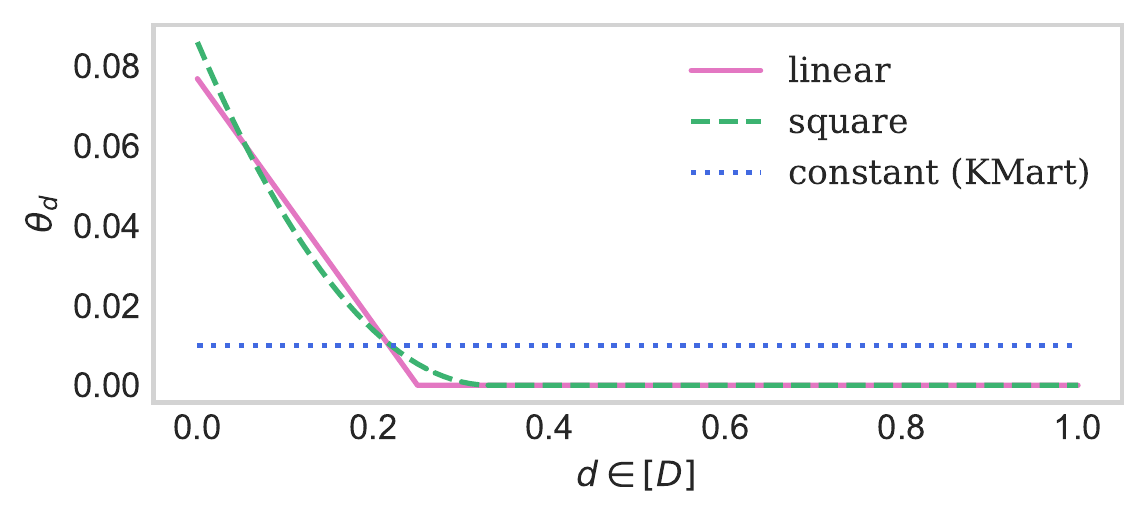}
    \caption{Various values of the convex weights $(\theta_1, \dots, \theta_D)$, which can be used in the construction of the diversified martingale \eqref{eq:diversifiedMartingale}. Notice that the linear and square weights are largest for $d$ near 0, and decrease as $d$ approaches $1/4$, finally remaining at 0 for all large $d$. Smaller values of $d$ are upweighted since they correspond to those values of $\lambda$ in $M_t^D(\mu^\star)$ that are optimal for smaller (i.e. interesting) electoral margins. This is in contrast to the constant weight function, which sets $\theta_d = 1/D$ for each $d \in [D]$. We find that square weights perform well in practice (see Figure~\ref{fig:stoppingTimes}) but these can be tuned and tailored based on prior knowledge and the particular problem at hand.}
    \label{fig:distKelly_distributions}
\end{figure}

Better choices of $(\theta_d)_{d=1}^D$ exist for the types of elections one might encounter in practice. Recall that near-optimal values of $\lambda$ are given by \eqref{eq:optimalLambda}. However, setting $\theta_d = 1/D$ for each $d \in [D]$ implicitly treats all $d / ((D+1) \cond_i(\mu^\star))$ as equally reasonable values of $\lambda$. Elections with large values of $\mu^\star$ (e.g. closer to 1) are ``easier'' to audit, and the interesting or ``difficult'' regime is when $\mu^\star$ is close to (but strictly larger than) 1/2. Therefore, we recommend designing $(\theta_1, \dots, \theta_D)$ so that $(M_t^D(1/2))_{t=0}^N$ upweights optimal values of $\lambda$ for margins close to 0, and downweights those for margins close to 1. Consider the following concrete examples. First, we have the truncated-square weights, 
\[ \theta_d^\mathrm{square} := \frac{\gamma_d^\mathrm{square}}{\sum_{d=1}^D \gamma^\mathrm{square}_d},~~~ \text{where } \gamma_d^\mathrm{square} := (1/3 - x)^2\1_{d \leq 1/3}. \]
and we normalize by $\sum_d \gamma_d^\mathrm{linear}$ to ensure that $\sum_d \theta_d = 1$.
Another sensible choice is given by the truncated-linear weights, where we simply replace $\gamma_d^\mathrm{square}$ by $\gamma_d^\mathrm{linear} := \max\{0, 1-2d\}$.
These values of $\theta_d^\mathrm{linear}$ and $\theta_d^\mathrm{square}$ are large for $d \approx 0$ and small for $d \gg 0$, and hence the summands in the martingale given by \eqref{eq:diversifiedMartingale} are upweighted for implicit values of $\lambda$ which are optimal for ``interesting'' margins close to 0, and downweighted for simple margins much larger than 0 (see Figure~\ref{fig:distKelly_distributions}). 

When $M_t^D$ is combined with $\theta_d^\mathrm{square}$, we refer to the resulting martingales and confidence sequences as \textbf{SqKelly}. We compare their empirical workload against that of \textit{a priori} Kelly, dKelly, and BRAVO in Figure~\ref{fig:stoppingTimes}. A hybrid approach is also possible: suppose we want to use reported outcomes or prior knowledge alongside these convex-weighted martingales. We can simply choose $(\theta_1, \dots, \theta_D)$ so that $M_t^D$ upweights values in a neighborhood of $\lambda'$ (or some other value chosen based on prior knowledge\footnote{The use of the word ``prior'' here should not be interpreted in a Bayesian sense. No matter what values of $(\theta_1, \dots, \theta_D)$ are chosen, the resulting tests and confidence sequences have \emph{frequentist} risk-limiting guarantees.}).

%% file: canada_audit.tex
\section{Illustration: Auditing Canada's 43rd Federal Election}
\label{section:canada}
We now apply the techniques derived in Section~\ref{section:design} to risk-limiting audits of the 2019 Canadian federal election, which is made up of many plurality contests between 6 major political parties.\footnote{While Canada has \href{https://www.elections.ca/content.aspx?section=pol&dir=par&document=index&lang=e}{many registered political parties}, only a handful have come close to winning seats in the house of commons, and hence should be considered in an audit. As a somewhat arbitrary rule, we considered those parties which satisfied the Leaders' Debates Commission's \href{https://debates-debats.ca/en/participation-criteria/interpretation-participation-criteria-leaders-debates/}{2019 participation criteria}.} These consisted of The Liberal Party of Canada, The Progressive Conservative Party of Canada (PC), The New Democratic Party (NDP), The Green Party, The Bloc Québécois (Bloc), and the People's Party of Canada (PPC). Independent candidates were also included where appropriate.
\begin{figure}[!htbp]
    \centering
    \includegraphics[width=0.6\textwidth]{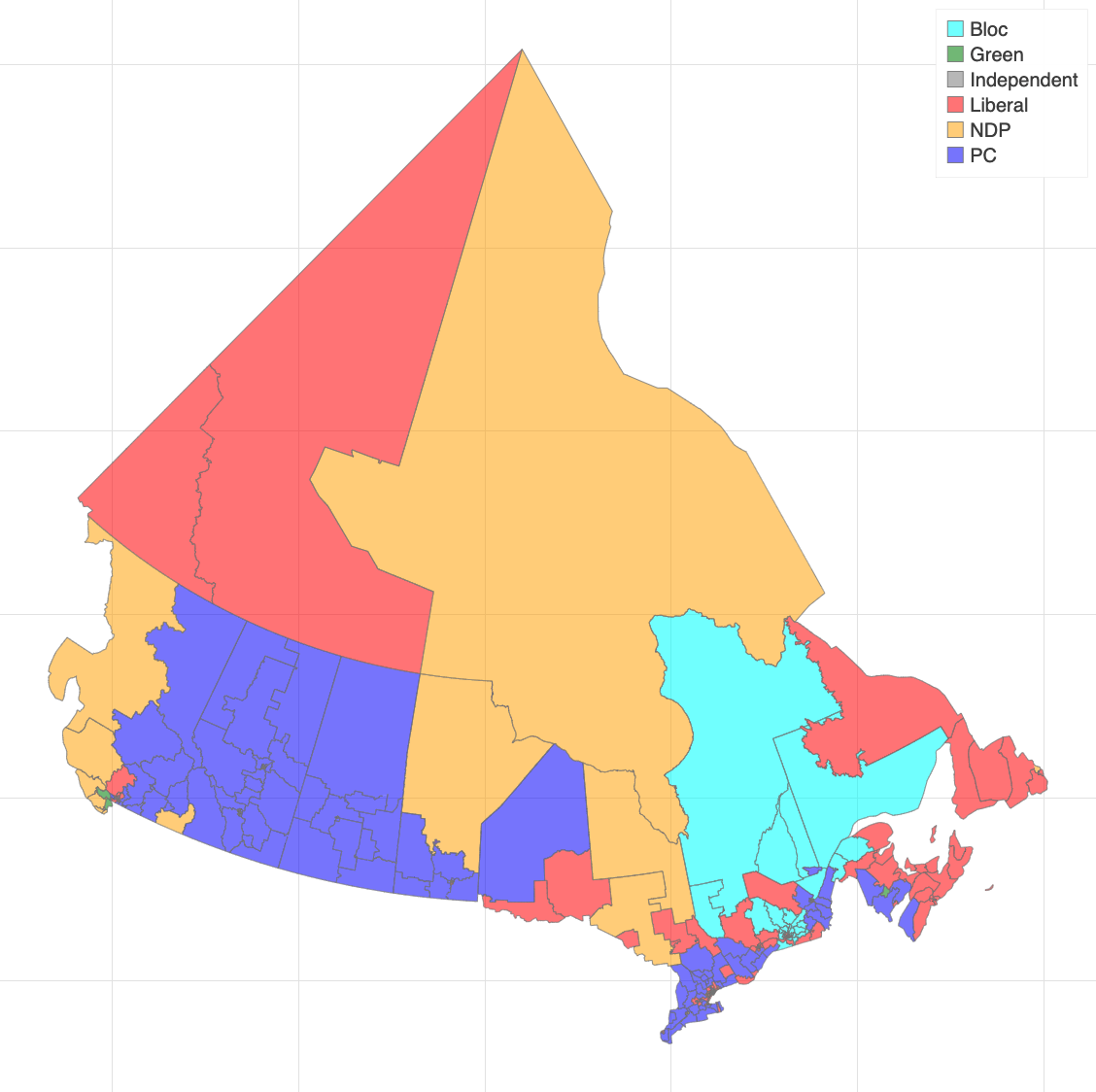}
    \caption{A map of Canada's 338 ridings, each representing one seat in the house of commons. Ridings are colored according to which party received the greatest number of votes in the 2019 federal election. The PPC is omitted from the legend here as they did not win any seats.}
    \label{fig:map}
\end{figure}
The country is made up of 338 so-called ``ridings'' (see Figure~\ref{fig:map}). These are geographic regions, each corresponding to one seat in the house of commons. For each riding, a multi-party, single-winner plurality contest takes place where the winner is awarded the respective seat. Generally speaking, the party with the greatest number of seats forms government (there are exceptions to this rule\footnote{\href{https://www.elections.ca/content.aspx?section=res&dir=ces&document=part1&lang=e}{www.elections.ca/content.aspx?section=res\&dir=ces\&document=part1\&lang=e}} but these will not be important for the purposes of auditing). In US elections, states and electoral college votes play similar roles to ridings and seats, respectively. Since each riding's underlying contest takes the form of a multi-party, single-winner plurality election, we can simply apply the techniques for auditing multiple contests outlined in Section~\ref{section:multipleContests} alongside the martingales and confidence sequences developed in Section~\ref{section:design}.

\paragraph{The data-driven web application}
We designed and developed an interactive Python- and Bokeh-based \cite{bokeh} web application where users can display audits of any Canadian riding in a single click. This combined two data sources: one for electoral outcomes as recorded by hand-counted paper ballots in the 2019 federal election \cite{elections_act_2019,canadaVotingResultsRaw}, and one to draw the map of electoral districts \cite{canadaElectoralDistricts}. After cleaning and merging, the data consisted of 347 records. Each record consists of a geographic information systems (GIS) polygon to draw the riding, vote totals for each party, and other information. The additional 9 records correspond to islands which are not separate ridings but require their own GIS polygon to be drawn on a map.

\begin{figure}[!htbp]
    \centering
    \includegraphics[width=\textwidth]{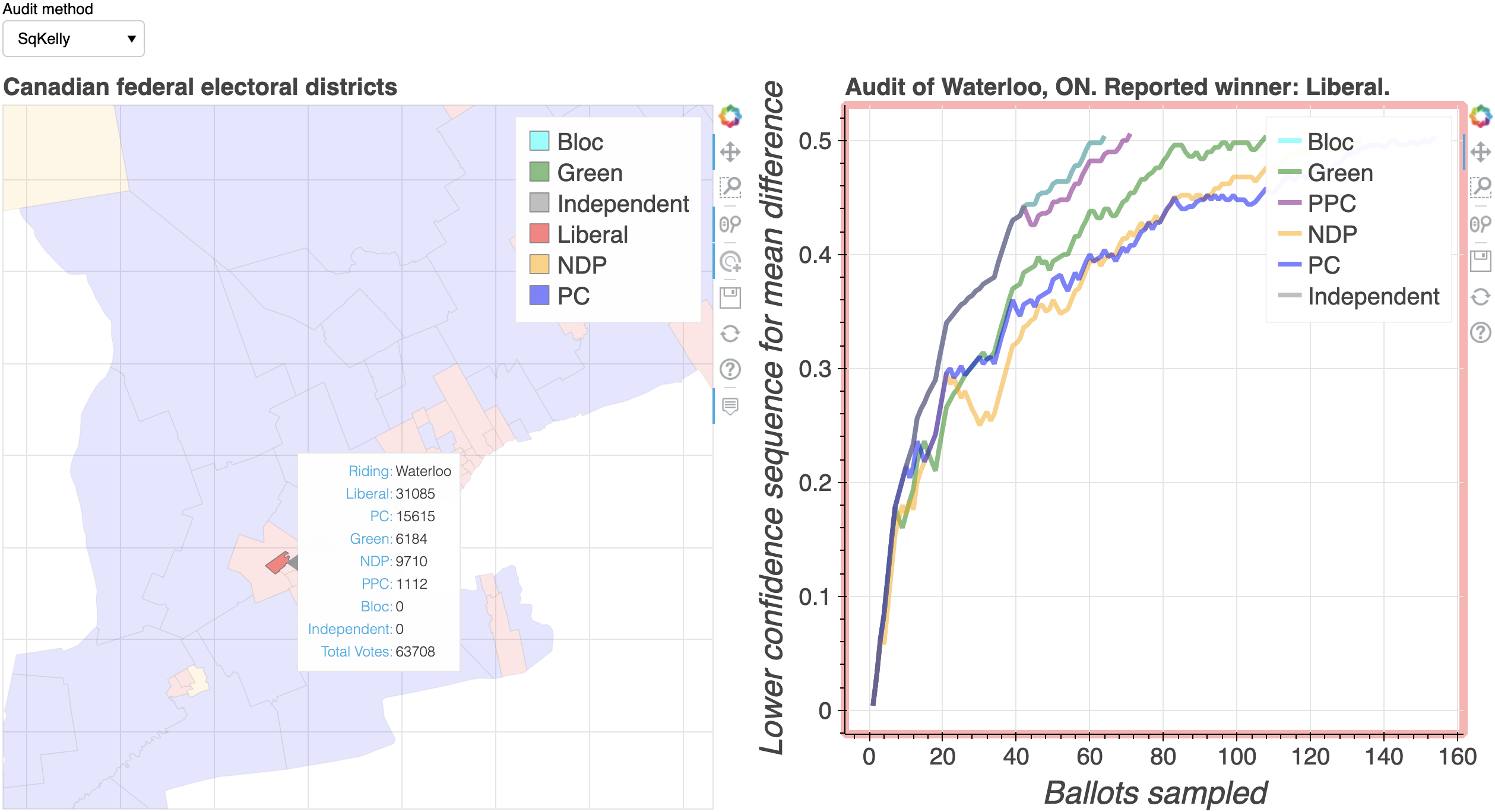}
    \caption{Example risk-limiting audit for the riding of Waterloo, Ontario using SqKelly. This screenshot was captured after zooming the map of Figure~\ref{fig:map} in on southern Ontario. In this example, it was (correctly) reported that the Liberal party received 31,085 out of 63,708 total votes. Clicking on Waterloo's polygon will begin the audit shown in the right-hand side, which displays six $(1-\alpha)$ lower confidence sequences for the pairwise contests between the Liberal party and each reportedly losing party. The Liberal party's win is certified once each of these confidence sequences exceeds 1/2, which in this case happened after sampling roughly 160 ballots.}
    \label{fig:example_audit}
\end{figure}
    
Following the notation of Section~\ref{section:multipleContests}, recall that the electoral parameter of interest $\mu_{w, \ell}$ is defined as  \[\mu^\star_{w, \ell} := \frac{1}{N} \sum_{i=1}^N x_i^{w, \ell}, \]
where
\begin{itemize}
    \item $x_i^{w, \ell} = 1$ if the $i^\text{th}$ ballot shows a vote for $w$,
    \item $x_i^{w, \ell} = 0$ if the $i^\text{th}$ ballot shows a vote for $\ell$, and
    \item $x_i^{w, \ell} = 1/2$ if the $i^\text{th}$ ballot shows a vote for any other party.
\end{itemize}
Also recall that the reported assertion --- ``$w$ received more votes than $\ell$ for each $\ell \in \mathcal L$'' --- is certified once the $(1-\alpha)$ lower confidence sequences for $\mu^\star_{w, \ell}$ exceed 1/2 for each $\ell \in \mathcal L$. Furthermore, this yields an RLA with risk limit $\alpha$, without needing to perform any multiplicity adjustments for constructing several confidence sequences (see Section~\ref{section:multipleContests} for more details). For example, the right-hand side plot of Figure~\ref{fig:example_audit} displays an RLA with risk-limit $\alpha$ for the assertion ``the Liberal party received the largest number of votes'' by computing six $(1-\alpha)$ lower confidence sequences for $\mu^\star_{w, \ell}$, where $w = \text{Liberal}$, and $\ell \in \{\text{Bloc, Green, PPC, NDP, PC, Independent}\}$.

It is important to keep in mind that electoral outcomes in the underlying data sets correspond to hand-counted paper ballot vote totals \cite{elections_act_2019, canadaVotingResultsRaw}. Therefore, the right-hand side plot in the web application (e.g. Figure~\ref{fig:example_audit}) demonstrates the length of time that an audit would last, given correctly-reported outcomes, \emph{and} assuming that the recorded data match the true votes cast. In practice, our confidence sequences would only rely on an assertion to audit (e.g. ``The Liberal party received the most votes'') and a simple random sample without replacement from the physical stack of ballots cast. Moreover, the web application is easily adapted to this practical scenario, an extension we plan to pursue in future work.

A key feature of this app is its interactivity. Users can hover their cursors over ridings to see reported vote totals, click and drag the map around, zoom in on regions of interest, and so on. When the user has found a riding they wish to audit, they can simply click on that riding's polygon to immediately compute lower confidence sequences and begin the RLA (see Figure~\ref{fig:example_audit}). Server-side computation and client-side updates are fully asynchronous, meaning users can interact with the app while the audit is being conducted, and the audit will not ``lock up''. A demo of these features can be found online\footnote{\href{https://ian.waudbysmith.com/audit\_demo.mov}{ian.waudbysmith.com/audit\_demo.mov}} and the code is available on GitHub.\footnote{\href{https://github.com/WannabeSmith/RiLACS/tree/main/canada_audit}{github.com/WannabeSmith/RiLACS/tree/main/canada\_audit}}

%% file: tallies.tex
\section{Risk-Limiting Tallies via Confidence Sequences}
\label{section:tallies}

Rather than audit an already-announced electoral outcome, it may be of interest to determine (for the purposes of making a first announcement) the election winner with high probability, without counting all $N$ ballots. 
Such procedures are known as risk-limiting tallies (RLTs), which were developed for coercion-resistant, end-to-end verifiable voting schemes \cite{jamroga2019risk}. For example, suppose a voter is being coerced to vote for Bob. If the final vote tally reveals that Bob received few or no votes, then the coercer will suspect that the voter did not comply with instructions. RLTs provide a way to mitigate this issue by providing high-probability guarantees that the reported winner truly won, leaving a large proportion of votes shrouded. In such cases, the voter is guaranteed plausible deniability, as they can claim to the coercer that their ballot is simply among the unrevealed ones.

While the motivations for RLTs are quite different from those for RLAs, the underlying techniques are similar. The same is true for confidence sequence-based RLTs. All  methods introduced in this paper can be applied to RLTs (with the exception of ``\textit{a priori} Kelly'' since it depends on the reported outcome) but with two-sided power. Consider the martingales we discussed in Section~\ref{section:designWithoutTotals},
\begin{equation}
\label{eq:dKellyRLT}
    M_t^D(\mu^\star) := \sum_{d=1}^D \theta_d \prod_{i=1}^t \left (1 + \frac{d}{(D+1) \cond_i(\mu^\star)} (X_i - \cond_i(\mu^\star)) \right),
\end{equation} 
where $(\theta_1, \dots, \theta_D)$ are convex weights. Recall that our confidence sequences at a given time $t$ were defined as those $\mu \in [0, 1]$ for which $M_t^D(\mu) < 1/\alpha$. In other words, a given value $\mu$ is only excluded from the confidence set if $M_t^D(\mu)$ is large. However, notice that $M_t^D(\mu)$ will become large if the conditional mean $\cond_t(\mu^\star) \equiv \EE(X_t \mid X_1, \dots, X_{t-1})$ is larger than the null conditional mean $\cond_t(\mu)$, but the same cannot be said if $\cond_t(\mu^\star) < \cond_t(\mu)$. As a consequence, the resulting confidence sequences are all one-sided \emph{lower} confidence sequences. To ensure that our bounds have non-trivial two-sided power, we can simply combine \eqref{eq:dKellyRLT} with a martingale that also grows when $\cond_t(\mu^\star) < \cond_t(\mu)$.
\begin{figure}[!htbp]
    \centering
    \includegraphics[width=0.7\textwidth]{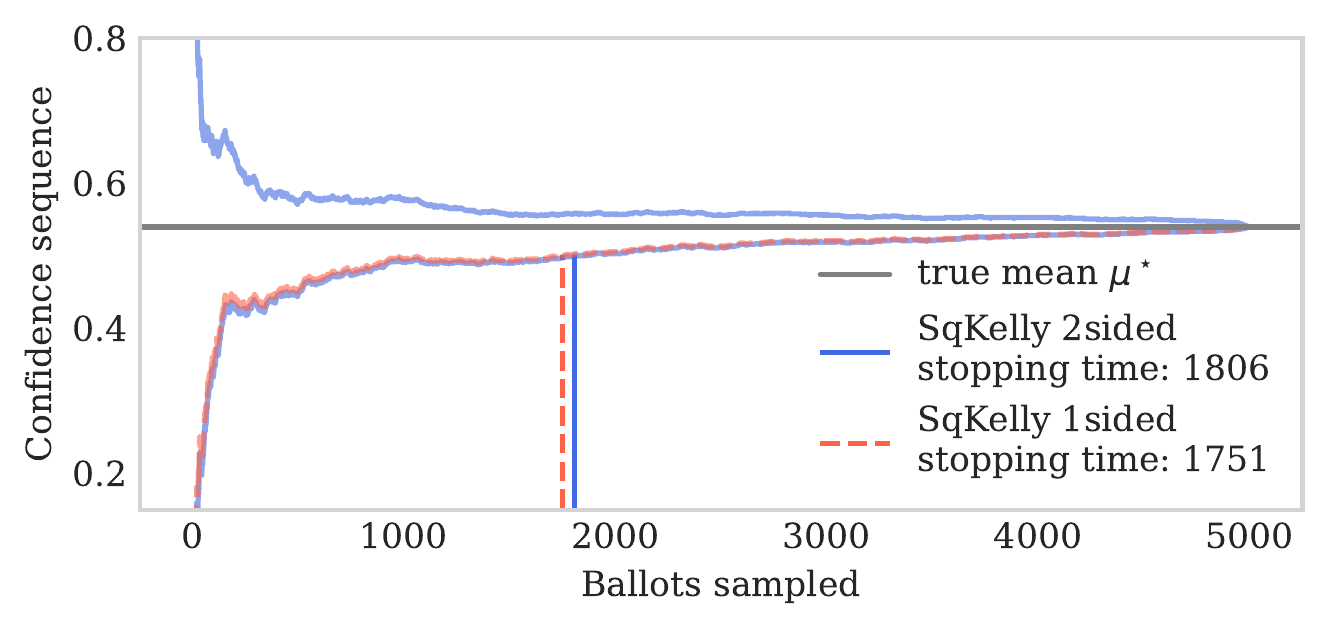}
    \caption{Confidence sequence-based risk-limiting tally for a two-candidate election. Unlike RLAs, RLTs require two-sided confidence sequences so that the true winner can be determined (with high probability) without access to an announced result. Notice that testing the same null $H_0: \mu^\star \leq 0.5$ is less efficient in an RLT than in an RLA. This is a necessary sacrifice for having nontrivial power against other alternatives.}
    \label{fig:RLT}
\end{figure}
\begin{proposition}
For nonnegative vectors $(\theta_1^+, \dots, \theta_D^+)$ and $(\theta_1^-, \dots, \theta_D^-)$ that each sum to one, define the processes
\begin{align*} 
M_t^{D+}(\mu) &:= \sum_{d=1}^D \theta_d^+ \prod_{i=1}^t \left (1 + \frac{d}{(D+1) \cond_i(\mu^\star)} (X_i - \cond_i(\mu^\star)) \right),\\
M_t^{D-}(\mu) &:= \sum_{d=1}^D \theta_d^- \prod_{i=1}^t \left (1 - \frac{d}{(D+1) (1-\cond_i(\mu^\star))} (X_i - \cond_i(\mu^\star)) \right). 
\end{align*}
Next, for $\beta \in [0, 1]$, define their mixture
\[ M_t^{D\pm}(\mu) := \beta M_t^{D+}(\mu) + (1-\beta) M_t^{D-}(\mu). \]
Then, $M_t^{D\pm}(\mu^\star)$ is a nonnegative martingale starting at one. Consequently, 
\[ C_t^\pm := \{\mu \in [0, 1] : M_t^{D\pm} (\mu) < 1/\alpha \} \]
forms a $(1-\alpha)$ confidence sequence for $\mu^\star$.
\end{proposition}
\begin{proof}
This follows immediately from the fact that both $M_t^{D+}(\mu^\star)$ and $M_t^{D-}(\mu^\star)$ are martingales with respect to the same filtration, and that convex combinations of such martingales are also martingales.
\end{proof}
With this setup and notation in mind, $M_t^{D}$ as defined in Section~\ref{section:designWithoutTotals} is a special case of $M_t^{D\pm}$ with $\beta = 1$. As noted by \cite{jamroga2019risk}, RLTs involving multiple assertions \emph{do} require correction for multiple testing, unlike RLAs. The same is true for confidence sequence-based RLTs (and hence the tricks of Section~\ref{section:multipleContests} do not apply). It suffices to perform a simple Bonferroni correction by constructing $(1 - \alpha/K)$ confidence sequences to establish $K$ simultaneous assertions.

%% file: conclusion.tex
\section{Summary}

This paper presented a general framework for conducting risk-limiting audits based on confidence sequences, and derived computationally and statistically efficient martingales for computing them. We showed how \textit{a priori} Kelly takes advantage of the reported vote totals (if available) to stop ballot-polling audits significantly earlier than extant ballot-polling methods, and how alternative martingales such as SqKelly also provide strong empirical performance in the absence of reported outcomes.
Finally, we demonstrated how a simple tweak to the aforementioned algorithms provides two-sided confidence sequences, which can be used to perform risk-limiting tallies. 
Confidence sequences and these martingales can be applied to ballot-level comparison audits and batch-level comparison audits as well, using ``overstatement assorters'' \cite{stark2019sets}, which reduce comparison audits to the same canonical statistical problem: testing whether the mean of any list in a collection of non-negative bounded lists is at most 1/2. We hope that this new perspective on RLAs and its associated software will aid in making election audits simpler, faster, and more transparent.

%% file: supplement.tex
\section{\texorpdfstring{Maximizing a proxy for $M_N(1/2)$}{Maximizing a proxy for the final product}}
\label{section:maximizeFinalWealth}
In Section~\ref{section:design}, equation \eqref{eq:finalWealth}, we considered the product
\begin{equation}
\label{eq:finalWealth2}
    \widetilde M_N^\lambda := \prod_{i=1}^N (1 + \lambda (x_i - 1/2) ),
\end{equation} 
as an (inexact) proxy for $M_N(1/2)$, the final value of the process $(M_t(1/2))_{t=0}^N$. Let us now show that the maximizer of $\widetilde M_N^\lambda$ is given by
    \begin{equation}
    \label{eq:argmaxFinalWealth}
        \lambda' := 2\frac{N_A' - N_B'}{N_A' + N_B'}. 
    \end{equation}
\begin{proof}
To begin, note that the maximizer of $\widetilde M_N^\lambda$ is exactly the maximizer of $\log(\widetilde M_N^\lambda)$ due to the monotinicity of $\log(\cdot)$. Taking the derivative of $\log(\widetilde M_N^\lambda)$ and setting it to zero, we find that $\widetilde M_N^\lambda$ is maximized by the value of $\lambda'$ that solves
\begin{equation}
\label{eq:logWealthMaximizerSolution}
    \sum_{i=1}^N \frac{x_i - 1/2}{1 + \lambda'(x_i - 1/2)} = 0.
\end{equation}
Breaking above sum up into terms for which the ballots are ones, zeros, and halves, respectively, we have that \eqref{eq:logWealthMaximizerSolution} reduces to
\begin{align*}
    0 &= \sum_{i : x_i = 1} \frac{x_i - 1/2}{1 + \lambda'(x_i - 1/2)} + \sum_{i : x_i = 0} \frac{x_i - 1/2}{1 + \lambda'(x_i - 1/2)} + \sum_{i : x_i = 1/2} \frac{x_i - 1/2}{1 + \lambda'(x_i - 1/2)} \\
    &= \sum_{i : x_i = 1} \frac{1/2}{1 + \lambda'/2} + \sum_{i : x_i = 0} \frac{- 1/2}{1 + -\lambda'/2} + \sum_{i : x_i = 1/2} \frac{0}{1 + \lambda' \cdot 0} \\
    &=\ N_A' \frac{1/2}{1 + \lambda'/2} - N_B' \frac{1/2}{1 - \lambda'/2}.
\end{align*} 
Solving the above equation for $\lambda'$ yields the desired result given in \eqref{eq:argmaxFinalWealth}. This completes the proof.
\end{proof}